\documentclass[11pt, letterpaper]{article}
\pdfoutput=1
\usepackage{amsmath}
\usepackage{theorem}
\usepackage{latexsym}
\usepackage{xspace}
\usepackage{amssymb}
\usepackage{fancybox}
\usepackage{epsfig}
\usepackage{full page, times}

\theoremheaderfont{\scshape}
\theorembodyfont{\slshape}
\newtheorem{theorem}{Theorem}
\newtheorem{definition}{Definition}
\newtheorem{lemma}[theorem]{Lemma}

\newtheorem{proposition}[theorem]{Proposition}

\theoremstyle{plain}
\theorembodyfont{\rmfamily}

\theorembodyfont{\itshape}

\newenvironment{proof}{\noindent {\sc Proof.  }}{$\Box$ \medskip}

 {
    \begin{enumerate}}{\end{enumerate}}

\newcommand\Ball{\mbox{Ball}}

\newcommand\reals{\mathbb{R}}

\newcommand{\marginlabel}[1]%
{\mbox{}\marginpar{\it{\raggedleft\hspace{0pt}#1}}}
\newcommand\poly{\mbox{poly}} 
\newcommand\quasipoly{\mbox{quasipoly}}  

\newcommand{\RR}{{\mathbb{R}}}

\newcommand\cP{\mathcal P}

\newlength{\pgmtab}  
\newtheorem{sublemma}[theorem]{Sublemma}
\newtheorem{conjecture}[theorem]{Conjecture}

\def\qed{ \ \vrule width.2cm height.2cm depth0cm\smallskip}
\def\eps{\epsilon}

\def\gSR{$(g, h)$-{\sc Sparse Regression}}

\def\nSR{$(g, h)$-{\sc Noisy Sparse Regression}}

\newcommand{\remove}[1]{}

\usepackage{amssymb}
\usepackage{fullpage}
\usepackage[noadjust]{cite}
\usepackage{amsfonts,amsmath,color,float,graphicx,verbatim,url}
\usepackage{enumerate}

\newcommand{\bfe}{\mathbf{e}}
\newcommand{\bfv}{\mathbf{v}}
\newcommand{\bfb}{\mathbf{b}}
\newcommand{\bfy}{\mathbf{y}}
\newcommand{\bfp}{\mathbf{p}}

\newcommand{\bfBB}{\mathbf{B}}
\newcommand{\bfB}{\mathbf{Ball}}
\newcommand{\bfu}{\mathbf{u}}
\newcommand{\bfw}{\mathbf{w}}
\newcommand{\bfz}{\mathbf{z}}
\newcommand{\bfc}{\mathbf{c}}
\newcommand{\bfx}{\mathbf{x}}
\newcommand{\bfeps}{\boldsymbol{\eps}}
\newcommand{\bfgamma}{\boldsymbol{\gamma}}
\newcommand{\bftheta}{\boldsymbol{\theta}}
\newcommand{\alg}{\mathcal{A}}

\newcommand{\cV}{\mathcal{V}}
\newcommand{\scNP}{\textsc{NP}}

\newcommand{\UA}{\textbf{UA}}
\newcommand{\SAT}{\textsc{SAT}}

\newcommand{\calB}{\mathcal{B}}
\newcommand{\calS}{\mathcal{S}}
\newcommand{\scBPP}{\textsc{BPP}}
\newcommand{\scBPTime}{\textsc{BPTime}}
\newcommand{\scDTime}{\textsc{DTime}}
\newcommand{\eat}[1]{}
\newcommand{\cost}{\text{cost}}
\newcommand{\PCP}{\textsc{PCP}}
\newcommand{\polylog}{\text{polylog}}

\begin{document}
\title{Variable Selection is Hard}
\author{Dean Foster
\thanks{Yahoo Labs, 229 West 43rd St., New York, NY 10036; {\tt dean@foster.net}.} \and
Howard Karloff
\thanks{Yahoo Labs, 229 West 43rd St., New York, NY 10036; {\tt karloff@yahoo-inc.com}.} \and
Justin Thaler
\thanks{Yahoo Labs, 229 West 43rd St., New York, NY 10036; {\tt jthaler@yahoo-inc.com}.}}

\date{$ $}
\maketitle

\sloppy
\begin{abstract}
Variable selection for sparse linear regression is the problem of finding,
given an $m\times p$  matrix $B$ and a target vector $\bfy$,  a sparse vector $\bfx$ such that $B\bfx$ approximately equals $\bfy$.
Assuming a standard complexity hypothesis, we show that no polynomial-time algorithm can find
a $k'$-sparse $\bfx$ with $||B\bfx-\bfy||^2\le h(m,p)$,
where $k'=k\cdot 2^{\log ^{1-\delta} p}$ and $h(m,p)\le p^{C_1} m^{1-C_2}$, where $\delta>0,C_1>0,C_2>0$ are arbitrary.  
This is true even under the promise that 
there is an unknown $k$-sparse vector $\bfx^*$ satisfying $B\bfx^*=\bfy$. 
We prove a similar
result for a statistical version of the problem in which the data are corrupted by noise.

To the authors' knowledge, these are the first hardness results for sparse regression that apply when the algorithm simultaneously has $k'>k$
and $h(m,p)>0$.
\end{abstract}


\section{Introduction}
\label{fullsec:intro}
Consider a linear regression problem in which one is given an $m \times p$ matrix $B$ and a target vector $\bfy \in \reals^m$.  The goal is to approximately represent $\bfy$ as a linear combination of
as few columns of $B$ as possible.  If a polynomial-time
algorithm $\alg$ is presented with $B$ and $\bfy$, and it is
known that $\bfy$ is an {\em exact} linear combination of some $k$ columns of $B$, 
and $\alg$ is allowed to choose {\em more than $k$} columns of $B$, how many columns must $\alg$ choose in order to generate a linear combination which is {\em close to} $\bfy$?

Note that we have allowed $\alg$ to ``cheat'' both on the number of columns and on the accuracy of the resulting linear combination.
In this paper, we show the problem is intractable despite the allowed cheating.

Formally, we define \gSR\ as follows.  Let $\bfe^{(m)}$ denote the $m$-dimensional vector of 1's, and for any vector $\bfz$, and let $\|\bfz\|_0$ denote the number of nonzeros in $\bfz$. Let $g(\cdot)$ and $h(\cdot, \cdot)$ denote arbitrary functions. An algorithm for \gSR\ satisfies the following.

\begin{itemize}
\item
Given:  An $m\times p$ Boolean matrix $B$ and a positive integer $k$ such that there is a real $p$-dimensional vector $\bfx^*$, $\|\bfx^*\|_0\le k$, such that $B\bfx^*=\bfe^{(m)}$.
(Call such an input {\em valid}.)
\item
Goal:  Output a (possibly random) $p$-dimensional vector $\bfx$ with $\|\bfx\|_0\le k\cdot g(p)$ such that $\|B\bfx-\bfe^{(m)}\|^2\le h(m, p)$ with high probability over the algorithm's internal randomness. 
\end{itemize}

Since we are focused on hardness, restricting $B$ to have entries in $\{0, 1\}$  and the target vector to be $\bfe^{(m)}$ only makes the result stronger.


There is, of course, an obvious exponential-time deterministic
algorithm for \gSR, for $g(p)=1$ for all $p$, and $h(m, p)=0$:  enumerate all subsets $S$ of $\{1,2,...,p\}$ of size $k$.
For each $S$, use Gaussian elimination to determine if there is an $\bfx$ with $\bfx_j=0$ for all $j\not\in S$ such that $B\bfx=\bfe^{(m)}$, and if there is one, to find it.
Since $B\bfx^*=\bfe^{(m)}$ and $\|\bfx^*\|_0\le k$, for at least one $S$ the algorithm must return an $\bfx$ with $\|\bfx\|_0\le k, B\bfx=\bfe^{(m)}$.

Before getting too technical, we warm up with a simple hardness
proof for \gSR.  The short proof can be understood on its own.
This theorem is just a warmup;  later we will prove a much
stronger result (relying, unfortunately, on a stronger complexity
assumption).  To the authors' knowledge, this simple proof
was not known, though similar arguments had been used previously to show weaker
hardness results for related problems
(cf. \cite[Proposition 6]{aroraetal} and
\cite{ipl}).)

\begin{theorem}\label{fullthm0}
Let $0<\delta<1$.
If there is a deterministic polynomial-time algorithm $A$ for \gSR, 
for which $g(p)=(1-\delta)\ln p$ and $h(m,p)=m^{1-\delta}$, 
then $\SAT\in \scDTime(n^{O(\log \log
n)})$.
\end{theorem}
\begin{proof}  Feige \cite{feige}
gives a reduction from \SAT,  running in deterministic time $N^{O(\log \log N)}$ on \SAT\ instances of size $N$,
to {\sc Set Cover}, in which the resulting, say, $m\times p$,
incidence matrix $B$ (whose rows
are elements and columns are sets) has the following properties.
There is a (known) $k$ such that
(1) if a formula $\phi\in \SAT$, then there is a collection of 
$k$ {\em disjoint} sets which covers the universe (that is, there
is a collection of $k$ columns of $B$ whose sum is $\bfe^{(m)}$), and (2)
if $\phi\not\in \SAT$, then no collection of at most $k\cdot[(1-\delta)\ln p]$ sets
covers the universe (in other words, no set of 
at most $k\cdot[(1-\delta)\ln p]$ columns of $B$ has a sum which is coordinate-wise at
least $\bfe^{(m)}$).
This is already enough to establish that any polynomial-time algorithm for $(g(p), 0)$-{\sc Sparse Regression} implies an 
$N^{O(\log \log N)}$-time algorithm for \SAT. The remainder of the proof is devoted
to establishing the analogous statement for $(g(p), m^{1-\delta})$-{\sc Sparse Regression} .

Build a matrix $B'$ by stacking $r$ copies of $B$ atop one another, $r$ to be 
determined later.    Let $M=rm$.  
If $\phi\in \SAT$, then there is a
collection of $k$ columns summing to $\bfe^{(M)}$.  This is a
linear combination of sparsity at most $k$.
If $\phi\not \in \SAT$, then for any linear combination of at most $k\cdot \left [ (1-\delta)\ln p\right]$ column vectors, in each of the $r$ copies of $B$, there is squared error 
at least 1 (since the best one could hope for, in the absence of a set cover, is $m-1$ 1's and one $0$).
This means that the squared error overall is at least $r$.  We want $r> M^{1-\delta}=(rm)^{1-\delta}$, i.e., $r^\delta>m^{1-\delta}$, and hence we define $r=\lceil 
m^{1/\delta-1}\rceil+1$.  

Construct an algorithm $A'$ for \SAT\ as follows.  
Run $A$ on instance $(B',k)$ of \gSR\  for $T(N)$ time steps, where $T(N)$ is the time $A$ would need on 
an $rm \times p$ matrix if $(B',k)$ were a valid input for \gSR\  (which it may not be);  
since in the valid case, $A$ runs in time polynomial in the size of the input matrix $B'$ (which is $N^{O(\log\log N)}$), 
$T(N)$ is also $N^{O(\log\log N)}$.
If $A$ outputs a vector $\bfx$ such that $\|B' \bfx - \bfe^{(M)}\|^2 \leq r$
and $\|\bfx\|_0 \leq k [(1-\delta) \ln p]$ within this time bound, then $A'$ outputs ``Satisfiable.'' Otherwise ($A$ doesn't terminate in the allotted time or 
it terminates and outputs an inappropriate vector), $A'$ outputs ``Unsatisfiable.''

Clearly $A'$ runs in time $N^{O(\log \log N)}$ on inputs $\phi$ of size $N$. 
It remains to show that $A'$
is a correct algorithm for \SAT. 

To this end, suppose that $\phi \in \SAT$. In this case, there is a solution $\bfx^*$ of sparsity at most $k$ with $B'\bfx^*=\bfe^{(M)}$,
and since $A$ is a correct algorithm for \gSR, $A$ will find such a solution, causing $A'$ to output ``Satisfiable'' when run on $\phi$.
On the other hand, if $\phi \not \in \SAT$, then there is no vector $\bfx^*$ with $\|\bfx^*\|_0 \leq k\cdot \left [ (1-\delta)\ln p\right]$
such that $\|B' \bfx^* - \bfe^{(M)}\|^2 \leq r$. Hence, $A'$ must output ``Unsatisfiable'' when run on $\phi$. 
We conclude that $A'$ is a correct algorithm for \SAT\ running in time $N^{O(\log \log N)}$ on instances of size $N$.
\end{proof}

One can combine Dinur and Steurer's PCP \cite{steurer} with Feige's construction, or the earlier
construction of Lund and Yannakakis \cite{lund}, 
to strengthen the conclusion of Theorem \ref{fullthm0} to $\SAT\in $P.

Throughout, $\scBPTime(T)$ will denote the set of all languages
decidable by randomized algorithms, with two-sided error, running in expected time $T$.
Our main result is that unless $NP\subseteq\scBPTime(n^{\polylog(n)})$, then even if $g(p)$ grows at a ``nearly polynomial'' rate, and $h(m, p) \leq p^{C_1} \cdot m^{1-C_2}$ for any positive constants $C_1, C_2$, there is no quasipolynomial-time (randomized) algorithm for \gSR:

\begin{theorem}\label{fullCSthm}
Assume that $\scNP \not \subseteq \scBPTime(n^{\polylog(n)})$.
For any positive constants $\delta, C_1, C_2$, there exist a
$g(p)$ in $2^{\Omega(\lg^{1-\delta}(p))}$ and an $h(m, p)$ in $\Omega\left(p^{C_1} \cdot m^{1-C_2}\right)$ 
such that there is no quasipolynomial-time randomized algorithm for \gSR.
\end{theorem}

Note that the ``$-C_2$'' cannot be removed from the condition $h(m, p) \leq p^{C_1}
\cdot m^{1-C_2}$ in Theorem \ref{fullCSthm}: 
the algorithm that always outputs the all-zeros vector solves \gSR\ for $g(p)=0$ and $h(m, p)= m$. 

We also show, assuming a slight strengthening of a standard
conjecture---known as the projection games conjecture (PGC)
\cite{projectiongames}---about the existence of
probabilistically checkable proofs with small soundness error,
that \gSR\ is hard even if $g$ grows as a constant power. We refer to
our slight strengthening of the PGC as the ``Biregular PGC,'' and 
state it formally in Section \ref{fullsec:projection}.

\begin{theorem}\label{fullCSthmstrong}
Assuming the Biregular PGC, the following holds: If $\scNP \not
\subseteq \scBPP$, then for any positive constants $C_1, C_2$, there 
exist a $g(p)$ in $p^{\Omega(1)}$ and an $h(m, p)$ in $\Omega\left(p^{C_1} \cdot m^{1-C_2}\right)$ such that there is no polynomial-time randomized algorithm for \gSR. \end{theorem}

We might consider \gSR\ to be a ``computer science'' version of Sparse Regression, in the sense that the data are deterministic and fully specified.
In an alternative, ``statistics" version of Sparse Regression, the data are corrupted by random noise unknown to the algorithm. 
Specifically we consider the  following problem, which we call \nSR.  


\begin{itemize}
\item
There are a positive integer $k$ and an $m\times p$ Boolean matrix $B$, such that there exists an unknown $p$-dimensional vector $\bfx^*$ with $\|\bfx^*\|_0\le k$ such that
$B\bfx^*=\bfe^{(m)}$.  
An $m$-dimensional vector $\bfeps$ 
of i.i.d. $N(0,1)$ ``noise'' components $\epsilon_i$ is generated and $\bfy$ is set to $B\bfx^*+\bfeps=\bfe^{(m)}+\bfeps$.  $B$, $k$, and $\bfy$ (but not $\bfeps$ or $\bfx^*$) are revealed to the algorithm.  
\item
Goal:  Output a (possibly random) $\bfx\in \reals^p$ such that
$E[\|B(\bfx-\bfx^*)\|^2]\le h(m, p)$  and $\|\bfx\|_0\le k\cdot g(p)$. Here, the expectation is taken over both the internal randomness of the algorithm and of the
$\epsilon_i$'s.
\end{itemize}

We give a simple reduction from \gSR\ to \nSR\ that proves the following theorems.
\begin{theorem}\label{fullstatthm}
Assume that $\scNP \not \subseteq \scBPTime(n^{\polylog(n)})$.
For any positive constants $\delta, C_1, C_2$, there exist a
$g(p)$ in $2^{\Omega(\log ^{1-\delta}(p))}$ and an $h(m, p)$ in
$\Omega(p^{C_1} \cdot m^{1-C_2})$ such that there is no
quasipolynomial-time randomized algorithm for \nSR.
\end{theorem}

\begin{theorem}\label{fullstatthmstrong}
Assuming the Biregular PGC, the following holds. If $\scNP \not \subseteq \scBPP$, then for any positive constants $C_1, C_2$, there 
exist a $g(p)$ in $p^{\Omega(1)}$  and an $h(m, p)$ in
$\Omega\left(p^{C_1} m^{1-C_2}\right)$ such that there is no
polynomial-time randomized algorithm for \nSR.
\end{theorem}

\paragraph{Importance and Prior Work.}
Variable selection is a crucial part of model design in statistics.  People want a model with a small number of variables partially for simplicity
and partially because models with fewer variables tend to have smaller generalization error, that is, they give better predictions on test data
(data not used in generating the model).  Standard greedy
statistical algorithms to choose features for linear regression include forward
stepwise selection (also known as stepwise regression or
orthogonal least squares),
backward elimination, and least angle regression. Standard non-greedy feature selection algorithms include LASSO and ridge regression.  

There are algorithmic results for sparse linear regression that guarantee good performance under certain conditions on the matrix $B$.  For example, 
equation (6) of \cite{ZWJ} states that a ``restricted eigenvalue condition" implies that the LASSO algorithm will give good performance for the statistical version of the problem.  
A paper by Natarajan \cite{Natarajan} presents an algorithm, also known as
forward stepwise selection or orthogonal least squares (see also
\cite{BD}), which achieves good performance
provided that
the $L_2$-norm of the pseudoinverse of the matrix obtained from $B$ by normalizing each column is small \cite{Natarajan}.
It appears that no upper bound for stepwise regression under
reasonable assumptions on the matrix was known prior to the
appearance of Natarajan's algorithm in 1995 (and the equivalence of
Natarajan's algorithm with forward stepwise selection appears not to
have been noticed until recently).
In the appendix, we include
an example proving that Natarajan's algorithm can perform badly
when the $L_2$-norm of that pseudoinverse is large, or  in other
words, that Natarajan's analysis of his algorithm is close to
tight.  Another example proving necessity of the factor involving the pseudoinverse
appeared in \cite[p. 10]{chen}.

There have been several prior works establishing hardness results for 
variants of the sparse regression problem. 
Natarajan \cite{Natarajan} used a reduction from {\sc Exact Cover By 3-Sets} to prove that,
 given an $m \times p$ matrix $A$, a vector $b \in \mathbb{R}^m$, and $\eps > 0$,
 it is NP-hard to compute a vector $x$ satisfying $\|Ax-b\| < \eps$ if such an $x$ exists, such that $x$ has the fewest
nonzero entries over all such vectors. Davis et al. \cite{davisetal} proved a similar NP-hardness result. 
The hardness results of \cite{Natarajan, davisetal} only establish hardness if the algorithm is
not allowed to ``cheat'' simultaneously on both the sparsity and accuracy of the resulting linear combination.

Arora et al. \cite{aroraetal} showed that, for any $\delta > 0$, a problem called \textsc{Min-Unsatisfy} does not have any polynomial-time algorithm achieving an
approximation factor of $2^{\log^{1-\delta}(n)}$, assuming $\textsc{NP} \not \subseteq \textsc{DTime}(n^{\polylog(n)})$. 
In this problem, the algorithm is given a system $Ax=b$ 
of linear equations over the rationals, and  the cost of a solution $x^*$ is the number of equalities that are violated by $x^*$. Amaldi and Kann \cite{amaldi} 
built directly on the result of Arora et al. \cite{aroraetal} to show, in our terminology, that $(2^{\log^{1-\delta}(n)}, 0)$-\textsc{Sparse-Regression} also has no polynomial-time algorithm
under the same assumption. Also see a result by S. Muthukrishnan
\cite[pp. 20-21]{muthu}. 

Finally, Zhang et al. \cite{ZWJ} showed a hardness result for \nSR.
We defer a discussion of the result
of \cite{ZWJ} to Section \ref{fullnoisymotivation}.  For now, we just note
that their hardness only applies to algorithms which cannot
``cheat'' on the sparsity, that is, to algorithms that must
generate a solution with at most $k$ nonzeros.

In summary, to the best of our knowledge, our work is the first to establish that 
sparse linear regression is hard to approximate, even when the algorithm
is allowed to ``cheat'' on \emph{both} the sparsity of the solution output and 
on the accuracy of the resulting linear combination.

\eat{
In light of the importance of variable (i.e., feature)
selection, it seems odd that so little is known regarding the computational complexity of finding
the best set of features.   Statisticians do not usually study the performance of their variable-selection algorithms in the worst case.  To our knowledge, the 
only paper giving a hardness result for a version of sparse
linear regression is \cite{ZWJ}. Their theorem involves the 
statistical version of the problem we call \nSR.  
Since we haven't yet defined the problem, we defer a discussion of the result
of \cite{ZWJ} to Section \ref{fullnoisymotivation}.  For now, we just note
that their hardness (unlike ours) only applies to algorithms which cannot
``cheat'' on the sparsity, that is, to algorithms that must
generate a solution with at most $k$ nonzeros.
}


\section{Proofs of Theorems \ref{fullCSthm} and \ref{fullCSthmstrong}}

\subsection{Notation and Proof Outline}
\label{fullsec:proofoverview}

Throughout, we will use lower-case boldface letters to denote vectors.
For any vector $\bfy \in \reals^m$, $\|\bfy\|$ will denote the Euclidean norm of $\bfy$, while
$\|\bfy\|_0$ will denote the sparsity (i.e., the number of
nonzeros) of $\bfy$. Let $\bfe^{(m)} \in \reals^m$ denote the
all-ones vector, and for any vector $\bfb$, let
$\bfB_{\Delta}(\bfb) = \{\bfb': \|\bfb - \bfb'\|^2 \leq
\Delta\}$ denote the ball of radius $\Delta$ around $\bfb$ in
the {\em square} of the Euclidean norm. 

If $\bfy = \sum_i c_i \bfw_i$ represents a vector $\bfy$ as a linear combination of vectors $\bfw_i$, we say that $\bfw_i$
\emph{participates} in
the linear combination if $c_i \neq 0$. 

We will use the symbol $N$ to denote the input size of $\SAT$ instances used in our reductions. 

The first step in our proof of Theorem \ref{fullCSthm} is to establish Proposition \ref{fullthm:weak} below.

\begin{proposition}
\label{fullthm:weak}
If $\SAT \not \in \scBPTime(N^{\polylog(N)})$, then for any constant $\delta < 1$,
there are a polynomial $m=m(p)$, a $k=k(p)$, and a pair 
$\sigma=\sigma(p), \Delta=\Delta(p)$ of values both in $2^{\Omega\left(\log^{1-\delta}(p)\right)}$,
such that no quasipolynomial-time randomized algorithm distinguishes the following two cases, given 
an $m \times p$ Boolean matrix $B$:

\begin{enumerate}
\item There is an $\bfx \in \{0, 1\}^p$ such that $B\bfx=\bfe^{(m)}$ and $\|\bfx\|_0 \leq k$.
\item For all $\bfx \in  \mathbb{R}^p$ such that $B\bfx \in
\bfB_{\Delta}(\bfe^{(m)})$, $\|\bfx\|_0 \geq k \cdot
\sigma$. 
\end{enumerate}
\end{proposition}
(For the purpose of proving Theorem \ref{fullCSthm}, having $\Delta=1$ in Proposition
\ref{fullthm:weak} would actually suffice.)

The second step of the proof describes a simple transformation of any \gSR\ algorithm for a ``fast-growing'' function $h$ into 
a \gSR\ algorithm for $h(m, p)=1$.
The proof appears in Section \ref{fullprop7}.

\begin{proposition}\label{fullthm:transform}
Let $C_1, C_2$ be any positive constants. Let $\alg$ be an
algorithm for \gSR\ running in time $T(m, p)$, for some function
$g(\cdot)\ge 1$, and for $h(m, p) = p^{C_1} m^{1-C_2}$. Then there is an algorithm $\alg'$ for $(g, 1)$-{\sc Sparse Regression}
that runs in time $\poly(T(\poly(m, p), p))$.
\end{proposition}

\begin{proof}[Proof of Theorem \ref{fullCSthm} assuming Propositions \ref{fullthm:weak} and \ref{fullthm:transform}]
Suppose by way of contradiction that there are positive
constants $\delta, C_1, C_2$ such that there is a
quasipolynomial-time randomized algorithm $\alg$ for \gSR\, where $g(p)=2^{\Omega(\log^{1-\delta}(p))}$ and $h(m, p)= p^{C_1} m^{1-C_2}$. 
By Proposition \ref{fullthm:transform}, $\alg$ can be transformed into a randomized quasipolynomial-time algorithm $\alg'$
for  $(g, 1)$-{\sc Sparse Regression}.

Clearly $\alg'$ is capable of distinguishing the following two cases, for any $\Delta \geq 1$:

\begin{enumerate}
\item There is an $\bfx \in \{0, 1\}^p$ such that $B\bfx=\bfe^{(m)}$ and $\|\bfx\|_0 \leq k$.
\item For all $\bfx \in  \mathbb{R}^p$ such that $B\bfx \in
\bfB_{\Delta}(\bfe^{(m)})$, $\|\bfx\|_0 \geq k \cdot
g(p)$.
\end{enumerate}

In particular, the above holds for $\Delta=\Delta(p)$ in $2^{\Omega(\log^{1-\delta}(p))} > 1$, which contradicts Proposition \ref{fullthm:weak}.
\end{proof}

\medskip \noindent \textbf{Proof Outline for Proposition \ref{fullthm:weak}.}
Lund and Yannakakis showed, assuming $\SAT$ cannot be
solved by algorithms running in time $O(N^{\text{polylog}(N)})$,
that $\textsc{Set-Cover}$ cannot be approximated within a factor of
$c \cdot \log_2 N$ for any constant $c < 1/4$ \cite{lund}. 
Here, an instance of $\textsc{Set-Cover}$ consists of a set $D$
of size $N$ and a family $\{D_1, \dots D_M\}$ of subsets of $D$, and the goal is to find a minimal collection of the $D_i$'s  whose union equals $D$. 

Lund and Yannakakis' transformation from an instance of $\phi$ of $\SAT$ to an instance of $\textsc{Set-Cover}$ has a (known) remarkable property: if $\phi$ is satisfiable, then the generated instance of
$\textsc{Set-Cover}$ does not just have a small cover
$\mathcal{C}$ of the base set $D$, it has a small
\emph{partition} of $D$. That is, if $\bfc_i$ denotes the
indicator vector of set $C_i$, then $\sum_{C_i \in \mathcal{C}}
\mathbf{c}_i = \bfe^{(|S|)}$. This is a stronger property than
$\sum_{C_i\in \mathcal{C}} \mathbf{c}_i\ge \bfe^{\left(|S|\right)}$, which is 
the condition required to show hardness of $\textsc {Set-Cover}$. 
This observation naturally allows us to define a corresponding instance of the Linear Regression problem with a sparse solution; the columns of the matrix $B$ in the regression problem are simply the indicator vectors $\bfc_i$.

%

A central
ingredient in Lund and Yannakakis' transformation is a certain
kind of set system $(V_1, \dots, V_M)$ over a base set $S$.
Their set system naturally requires that any union of fewer than
$\ell$ $V_i$'s or $\overline{V_i}$'s cannot \emph{cover} $S$, unless
there is some $i$ such that both $V_i$ and $\overline{V_i}$
participate in the union. As a result, they needed $|S|$ to be
polynomial in $M$ and exponential in $\ell$. Since we are studying Sparse Regression rather than $\textsc{Set-Cover}$, we can impose a weaker condition on our set systems. Specifically, we need that: 
\begin{center}
any \emph{linear combination} of $\ell$ indicator vectors of the
sets or their complements is ``far'' from $\bfe^{(m)}$, unless the
linear combination includes both the indicator vector of a set
and the indicator vector of its complement. 
\end{center}
(See Definition \ref{fulldef:setsystem} below.) As a result, we are able to take $|S|$ to be much smaller relative to $M$ and $\ell$ than can Lund and Yannakakis. Specifically, we can take $|S|$ to be \emph{polynomial} in $\ell$ and \emph{logarithmic} in $M$. This is the key to establishing hardness for super-logarithmic approximation ratios, and to obtaining hardness results even when we only require an approximate solution to the system of linear equations.

\subsection{Proof of Proposition \ref{fullthm:weak}}

\subsubsection{Preliminaries}
A basic concept in our proofs is that of 
\emph{$\Delta$-useful set systems}, defined below.

\begin{definition}
\label{fulldef:setsystem}
Let $M$ and $\ell$ be positive integers, and let $S$ be any finite set.
A \emph{set system}  $\calS_{M, \ell}=\{V_1, \dots, V_{M}\}$ 
of size $M$ over $S$ 
is any collection of $M$ distinct subsets of $S$.  
 $\calS_{M, \ell}$ is called \emph{$\Delta$-useful}, for $\Delta\ge 0$, if the following properties are satisfied. 
 
 Let $\bfv_i \in \{0, 1\}^{|S|}$ denote the indicator vector of $V_i$, that is, $\bfv_{i,j}=1$ if $j \in V_i$, and 0 otherwise. Let $\bar\bfv _i$ denote the indicator vector of the complement of $V_i$. 
 Then no $\ell$-sparse linear combination of the vectors 
 $\bfv_1, \bar\bfv_1, \dots, \bfv_M, \bar\bfv_M$ is in $\bfB_{\Delta}(\bfe^{\left(|S|\right)})$, unless there is some $i$
 such that $\bfv_i$ and $\bar \bfv_i$ both participate in the linear combination.
\end{definition}

(Note that there is a 2-sparse linear combination involving
$\bfv_i$ and $\bar\bfv_i$ that exactly equals
$\bfe^{\left(|S|\right)}$, namely, $\bfv_i+\bar
\bfv_i=\bfe^{\left(|S|\right)}$.)

\begin{lemma}\label{fulllemma:setsystem}
For any pair $M,\ell$, $M\ge 2$, of positive integers, there exists a
set $S=\{1,2,...,|S|\}$ of size $O(\ell^2 \cdot \log M)$ such
that there is a $\Delta$-useful set system $\calS_{M, \ell}$
over $S$, for some $\Delta=\Omega(|S|)$. Moreover, there is a
polynomial-time randomized algorithm that takes $M$ and $\ell$
and generates a $\Delta$-useful set system $\calS_{M, \ell}$ over $S$
with probability at least .99. 
\end{lemma}

(It seems likely that a deterministic construction that appears in \cite{aroraetal} could be modified to generate a $\Delta$-useful set system.)

\begin{proof}
Throughout the proof, we set 
\begin{equation}\label{fulleq:Sdef} |S| = \lceil 256 \ell^2 \ln M\rceil. \end{equation}
To avoid notational clutter, we denote $\bfe^{(|S|)}$ simply as
$\bfe$ throughout the proof of Lemma \ref{fulllemma:setsystem}. 
The core of our argument is the following technical lemma bounding the probability that $\bfe$ is ``close''
to the span of a ``small'' number of randomly chosen vectors from $\{0, 1\}^{|S|}$. 

\begin{sublemma}\label{fullthesublemma}
Given $M\ge 2$ and $\ell$, define $|S|$ as above.
Let $\bfv_1, \dots, \bfv_{\ell} \in \{0, 1\}^{|S|}$ be chosen independently and uniformly at random from $\{0, 1\}^{|S|}$.
Let $E$ denote the event that there exist coefficients $c_1,
\dots, c_{\ell} \in \reals$ such that $\|\bfe -
\left(\sum_{i=1}^{\ell} c_i \bfv_i\right)\|^2 \leq |S|/32$.
Then the probability
of $E$ is at most $M^{-32\ell}$. 
\end{sublemma}
\begin{proof}
Rather than reasoning directly about the Boolean vectors $\bfv_1, \dots, \bfv_\ell$ in the statement of the sublemma,
it will be convenient to reason about vectors in $\{-1, 1\}^{\ell}$. Accordingly, for any vector $\bfv \in \{0, 1\}^{|S|}$,
define $\bfv^*= 2\bfv-1$. Notice that if $\bfv$ is a uniform random vector in $\{0,1\}^{|S|}$, then $\bfv^*$ is a uniform random vector in $\{-1, 1\}^{|S|}$. The primary reason we choose
to work  with vectors over $\{-1, 1\}^{|S|}$ rather than $\{0, 1\}^{|S|}$ is that 
vectors in $\{-1, 1\}^{|S|}$ always have squared Euclidean norm exactly equal to $|S|$, in contrast to Boolean vectors, which can have squared Euclidean norm  as low as 0 and as large as
$|S|$. Throughout, given a set $T$ of vectors in $\reals^{|S|}$, and another vector $\bfv$ in $\reals^{|S|}$,
we let $\Pi_T(\bfv)$ denote the projection of $\bfv$ onto the linear span of the vectors in $T$.

Note that for any Boolean vectors $\bfw, \bfv \in \{0, 1\}^{|S|}$, 
$\|\bfw^*-\bfv^*\|^2=
4\|\bfw - \bfv\|^2 $. 
Hence, event $E$ from the statement of Sublemma \ref{fullthesublemma} occurs if and only if there exist coefficients $c_1, \dots, c_{\ell} \in \reals$ such that $\|\bfe^* - \left(\sum_{i=1}^{\ell} c_i \bfv^*_i\right)\|^2 \leq |S|/8$. We bound the probability of this occurrence as follows.

Let $T = \{\bfv^*_1, \dots, \bfv^*_{\ell}\}$. Note that
\[\min_{c_1, \dots, c_{\ell} \in \reals} \|\bfe^* - \left(\sum_{i=1}^{j} c_i \bfv^*_i\right)\|^2 = \|\bfe^*\|^2 - \|\Pi_T(\bfe^*)\|^2 = |S| - \|\Pi_T(\bfe^*)\|^2.\]
Combined with the previous paragraph, we see that event $E$ occurs if and only if 
$$\|\Pi_T(\bfe^*)\|^2 \ge 7|S|/8.$$ 
By equation \eqref{fulleq:Sdef}, $7|S|/8  > 128 \ell^2
\ln M,$ so it suffices to bound from above the probability that
$$\|\Pi_T(\bfe^*)\|^2  > 128 \ell^2 \ln M.$$

\paragraph{Bounding the probability that
$\|\Pi_T(\bfe^*)\|^2  > 128 \ell^2 \ln M.$}
Our analysis consists of two steps. In step 1, 
we
 bound the probability that $\|\Pi_T(\bfw^*)\|^2  > 128 \ell^2 \ln M$, where $\bfw^*$
is a vector chosen uniformly at random from $\{-1, 1\}^{|S|}$, independently of $\bfv^*_1, \dots, \bfv^*_\ell$. 
In step 2, we argue that, even though $\bfe^*$
is a fixed vector (not chosen uniformly at random from $\{-1, 1\}^{|S|}$), it still holds that
$\|\Pi_T(\bfe^*)\|^2  > 128 \ell^2 \ln M$ with exactly the same probability
(where the probability is now only over the random choice of vectors $\bfv^*_1, \dots, \bfv^*_{\ell} \in \{-1, 1\}^{|S|}$).

\paragraph{Step 1.} Let $\ell' \leq \ell$ denote the dimension
of the linear subspace $V:=\text{span}\left(\bfv^*_1, \dots, \bfv^*_\ell\right)$. 
Let $\{\bfz^*_1, \dots, \bfz^*_{\ell'}\}$ denote an arbitrary orthonormal basis for $V$. 
We have 
\begin{equation}\label{fulleqn2}
\|\Pi_T(\bfw^*)\|^2 = \sum_{i=1}^{\ell'} \langle \bfw^*, \bfz_i \rangle^2 = \sum_{i=1}^{\ell'} \left(\sum_{j=1}^{|S|} \bfw^*_j \cdot \bfz_{i, j} \right)^2,
\end{equation}
where $\bfw^*_j$ denotes the $j$th entry of $\bfw^*$, and $\bfz_{i, j}$ denotes the $j$th entry of $\bfz$. For any $i$,
let $w_i$ denote the value of the sum  
$w_i = \sum_{j=1}^{|S|} \bfw^*_j \cdot \bfz_{i, j}$. 
Since $\bfz_i$ is a vector in an orthonormal 
basis, we know that $\sum_{j=1}^{|S|} \bfz_{i, j}^2 = 1$.

Meanwhile, each $\bfw^*_j$ is a Rademacher random variable.
Because we can view $\bfv^*_1,...,\bfv^*_\ell$ and hence
$\bfz^*_1,...,\bfz^*_{\ell'}$ as fixed while 
$\bfw^*$ is random, 
and because $\sum_j \bfz_{i,j}^2=1$, 
standard concentration results for Rademacher sequences
\cite{rademacher} imply
that for $t>0$ and for all $i$,
$\Pr[w_i > t] 
=\Pr[ \sum_{j=1}^{|S|} \bfw^*_j \cdot \bfz_{i, j}>t]
\leq e^{-t^2/2}$, for all fixed 
$\bfv^*_1,...,\bfv^*_\ell$, and hence even if 
$\bfv^*_1,...,\bfv^*_\ell$ is random.
In particular, $\Pr[w_i > \sqrt{128 \ell \ln M}] \leq e^{-64 \ell \ln M}$. 
A union bound over all $i$ implies that $\Pr[w_i > \sqrt{128
\ell \ln M} \text{ for all } i \in [\ell']] \leq \ell' \cdot e^{-64 \ell \ln M} \leq e^{-32 \ell \ln M} = M^{-32 \ell}$.
In this event,  
i.e., for all $i$ $\sum_{j=1}^{|S|} \bfw^*_j \cdot \bfz_{i,
j}\le \sqrt{128 \ell \ln M}$, 
we can bound the right-hand side of equation (\ref{fulleqn2}) by
\[
\label{fulleq:done} \sum_{i=1}^{\ell'} \left(\sum_{j=1}^{|S|}
\bfw^*_j \cdot \bfz_{i, j} \right)^2 \leq \ell' \cdot 128 \ell
\ln M \leq 128 \ell^2 \ln M.
\]

\paragraph{Step 2.}
For vectors $\bfv^*_1, \dots, \bfv^*_{\ell}, \bfw^* \in \{-1, 1\}^{|S|}$, 
let $\mathbb{I}(\bfv^*_1, \dots, \bfv^*_{\ell}; \bfw^*)$ equal 1
if $\|\Pi_T(\bfw^*)\|^2  > 128 \ell^2 \ln M$, and 0 otherwise, where $T=\{\bfv^*_1, \dots, \bfv^*_{\ell}\}$. 
We claim that $\mathbb{E}[\mathbb{I}(\bfv^*_1, \dots, \bfv^*_{\ell}; \bfw^*)] = \mathbb{E}[\mathbb{I}(\bfv^*_1, \dots, \bfv^*_{\ell}; \bfe^*)],$
where the first expectation is over random choice of $\bfv^*_1, \dots, \bfv^*_{\ell}, \bfw^* \in \{-1, 1\}^{|S|}$, and
the second expectation is only over the random choice of $\bfv^*_1, \dots, \bfv^*_{\ell} \in \{-1, 1\}^{|S|}$.

For two vectors $\bfx, \bfy \in \{-1, 1\}^{|S|}$, let $\bfx
\otimes \bfy$ denote the component-wise product of $\bfx$ and
$\bfy$, i.e.,
$(\bfx \otimes \bfy)_i = \bfx_i \cdot \bfy_i$. Then
we can write:
\begin{flalign*} \mathbb{E}[\mathbb{I}(\bfv^*_1, \dots, \bfv^*_{\ell}; \bfw^*)]
& =  2^{-|S| \cdot (\ell+1)} \sum_{\bfv^*_1, \dots, \bfv^*_{\ell}, \bfw^* \in \{-1, 1\}^{|S|}} \mathbb{I}(\bfv^*_1, \dots, \bfv^*_{\ell}; \bfw^*)\\
& =  2^{-|S| \cdot (\ell+1)} \sum_{\bfw^* \in \{-1,
1\}^{|S|}} \left [\sum_{\bfv^*_1, \dots, \bfv^*_{\ell} \in \{-1,
1\}^{|S|}} \mathbb{I}(\bfv^*_1, \dots, \bfv^*_{\ell};
\bfw^*)\right ]\\
& =  2^{-|S| \cdot (\ell+1)} \sum_{\bfw^* \in \{-1,
1\}^{|S|}} \left [\sum_{\bfv^*_1, \dots, \bfv^*_{\ell}  \in
\{-1, 1\}^{|S|}} \mathbb{I}(\bfv^*_1 \otimes \bfw^*, \dots,
\bfv^*_{\ell} \otimes \bfw^*; \bfw^*\otimes\bfw^*)\right ]\\
& =  2^{-|S| \cdot \ell} \sum_{\bfv'_1, \dots,
\bfv'_{\ell} \in \{-1, 1\}^{|S|}} \mathbb{I}(\bfv'_1, \dots,
\bfv'_{\ell}; \bfe^*)\\
&=  \mathbb{E}[\mathbb{I}(\bfv'_1, \dots, \bfv'_{\ell}; \bfe^*].
\end{flalign*}
Here, the first equality is the definition of expectation; the
second follows by rearranging the sum.
The third equality holds because multiplying each of
$\bfv^*_1,...,\bfv^*_{\ell}$ and $\bfw^*$ component-wise by $\bfw^*$
is the same as premultiplying each vector by the {\em unitary}
matrix $U=diag(\bfw^*)$, and premultiplying by a unitary matrix
just rotates the space, which doesn't change lengths of
projections.
The fourth equality holds because, for any fixed $\bfw^*$, if the vectors $\bfv^*_1, \dots, \bfv^*_{\ell}$ are uniformly distributed
on $\{-1, 1\}^{|S|}$, then so are 
the vectors $\bfv'_1:=\bfv^*_1 \otimes \bfw^*, \dots,
\bfv'_\ell:=\bfv^*_{\ell} \otimes \bfw^*$.
The final equality is the definition of expectation.
\end{proof}

Let $V_1, \dots, V_{M}$ be random subsets of $S$, and $\bfv_i$ be the indicator vector of $V_i$. 
Consider any subset $Z$ of $\ell$ of the vectors  $\bfv_1, \bar\bfv_1, \dots, \bfv_M, \bar\bfv_M$ in which
there is no $i$ such that both $\bfv_i$ and $\bar\bfv_i$ are in $Z$. (There are exactly ${M\choose \ell} 2^\ell$ such subsets $Z$.) Then the vectors in $Z$ are all uniformly
random vectors in $\{0, 1\}^{|S|}$ that are chosen independently of each other.
Sublemma \ref{fullthesublemma}
implies that the probability that 
there exist coefficients $c_1, \dots, c_{\ell} \in \reals$ such that $\|\bfe - \left(\sum_{i=1}^{\ell} c_i \bfv_i\right)\|^2 \leq |S|/32$ is at most $M^{-32\ell}$.

We can take a union bound over all ${{M \choose \ell}}$ $\cdot 2^{\ell} \le M^{2\ell}$
possible choices of $Z$ to conclude that 
$\|\bfe - \left(\sum_{i=1}^{\ell} c_i \bfv_i\right)\|^2 > |S|/32$ for \emph{all} possible choices of the subset $Z$ with probability
at least $1-M^{2\ell} \cdot M^{-32\ell}=1-M^{-30\ell}\ge 0.99$.  
In this event, there is no $\ell$-sparse linear combination of the $\bfv_i$ and $\bar\bfv_i$ vectors in 
$\Ball_{\Delta}(\bfe)$ for $\Delta = |S|/32$, unless $\bfv_i$ and $\bar\bfv_i$ both participate in the linear combination for some $i$.
That is, $\{V_1, \dots, V_M\}$ is a  $\Delta$-useful set system, as desired.

This completes the proof of Lemma \ref{fulllemma:setsystem}.
\end{proof}

\subsubsection{Full Proof of Proposition \ref{fullthm:weak}}

Suppose that $\alg'$ is a $\scBPP$ algorithm that distinguishes the two cases appearing in the statement of Proposition \ref{fullthm:weak}. We show how to use $\alg'$ to design an efficient randomized algorithm for $\SAT$. 

\medskip
\noindent \textbf{MIP Notation.} 
Our construction of $\alg'$ assumes the existence of a perfectly complete multiprover interactive proof (MIP) with two provers and soundness error $\eps_{\text{sound}}<1$. 
Let $R$ denote the possible values of the MIP verifier's private randomness, $Q_i$ denote the set of possible 
queries that may be posed to the $i$th prover, and $A_i$ denote the set of possible responses from the $i$th prover.
Denote by $\cV(r, a_1, a_2) \in \{0, 1\}$ the verifier's
output given internal randomness $r$ and answers $(a_1, a_2)$ from the two provers, with an output of 1 being interpreted
as the verifier's accepting the provers' answers as valid proofs that the input is satisfiable.

As in prior work extending Lund and Yannakakis' methodology \cite{bellare}, we require the MIP to satisfy 
four additional properties, and we call an MIP \emph{canonical} if all four additional properties hold. Specifically: 

\label{fullsec:functionality}
\begin{definition}
\label{fulldef:canonical}

An MIP is said to be \emph{canonical} if it satisfies the following four properties. 
(These properties are, of course, independent of any provers.)
\begin{itemize}
\item \emph{Functionality}: for each $r \in R$ and each $a_1 \in
A_1$, there is at most one $a_2 \in A_2$
with the property that the verifier, after choosing $r\in R$, accepts answer vector $(a_1, a_2)$.

\item \emph{Uniformity}: For each $i\in \{1,2\}$, the verifier's queries to $\cP_i$ are distributed uniformly over $Q_i$.

\item \emph{Equality of Question-Space Sizes}: the sets $Q_1$ and $Q_2$ are of the same size.

\item \emph{Disjointness of Question and Answer Spaces}: $A_1$ and $A_2$ are disjoint, as are $Q_1$ and $Q_2$. 

\end{itemize}
\end{definition}

We will show how to turn any canonical MIP for $\SAT$ into a randomized algorithm $\alg''$ that accomplishes the following.
Fix any function $\ell: \mathbb{N} \rightarrow \mathbb{N}$.
Given any $\SAT$ instance $\phi$ of size $N$, $\alg''$ 
outputs an $m \times p$ Boolean matrix $B$,
where $m$ and $p$ are polynomial in the parameters 
$|R|, |Q_1|, |Q_2|, |A_1|, |A_2|$. 
Moreover, $B$ satisfies the following two properties 
with high probability over the internal randomness of $\alg''$.

\begin{itemize}
\item[Property 1:] If $\phi \in \SAT$, then there is a vector $\bfx^* \in \{0, 1\}^p$ such that $B \bfx^*=\bfe^{(m)}$ and 
$\|\bfx^*\|_0 \leq |Q_1| + |Q_2|$. 
\item[Property 2:] 
If $\phi \not\in \SAT$, then any $\bfx \in \mathbb{R}^p$ such that $B\bfx \in \bfB_{\Delta}(\bfe^{(m)})$
satisfies $\|\bfx\|_0 \geq \left [(1-\eps_{\text{sound}} \cdot \ell(N)^2) \cdot \frac{\ell(N)}{2}\right ] \cdot \left( |Q_1| + |Q_2|\right)$,
for some $\Delta = \Omega\left(\ell^2(N) \cdot \log \left( |A_1| + |A_2| \right)\right)$.
\end{itemize}

Before describing $\alg''$, we complete the proof of Proposition \ref{fullthm:weak}, assuming the existence of $\alg''$.

\medskip 
\begin{proof}
It is well-known that the standard PCP theorem can be combined
with Raz's parallel repetition theorem \cite{raz}  to yield a
canonical two-prover MIP for a $\SAT$ instance of size $N$ with the following two properties, for any constant $c > 0$: 

\begin{itemize}
\item the soundness error is $\eps_\text{sound}(N) \leq 2^{-\log^{c}(N)}$, and

\item $|R|, |A_1|, |A_2|, |Q_1|, |Q_2|$ are all of size $N^{O(\log^{c}(N))}$.
\end{itemize}
\medskip
See, e.g., \cite[Section 2.2]{feige}.
Throughout the remainder of the proof, we choose $c$ to satisfy$\frac{c}{1+c} = 1-\delta$.  

Hence, setting $\ell(N) = (10\cdot
\eps_{\text{sound}}(N))^{-1/2}$ and $k=|Q_1| + |Q_2|$, properties 1 and 2 above imply the following.
Given an instance $\phi$ of $\SAT$ of size $N$, $\alg''$ outputs an $m \times p$ Boolean matrix $B$,
where $m$ and $p$ are $N^{O(\log^{c}(N))}$, and $B$ satisfies the following properties. 
\begin{itemize}
\item If $\phi \in \SAT$, then there is a vector $\bfx^* \in \{0, 1\}^p$ such that $B \bfx^*=\bfe^{(m)}$ and 
$\|\bfx^*\|_0 \leq k$. 
\item If $\phi \not\in \SAT$, then any $\bfx \in \mathbb{R}^p$ such that $B\bfx \in \bfB_{\Delta}(\bfe^{(m)})$
satisfies $\|\bfx\|_0 \geq \sigma \cdot k$, where $\ell=\ell(N)$ and $\Delta = \Omega\left(\ell^2 \cdot \log \left( |A_1| + |A_2| \right)\right) \geq \Omega(\ell)$
and $\sigma = (1-\eps_{\text{sound}} \cdot \ell^2) \cdot \frac{\ell}{2} \geq 0.45 \ell$.
\end{itemize}

Suppose that there is a randomized algorithm $\alg'$ that runs in $\scBPTime(\quasipoly(m, p)) = \scBPTime(N^{\polylog(N)})$ and distinguishes 
the above two cases. By running $\alg'$ on the matrix $B$ output by $\alg''$, we determine with high probability whether 
$\phi$ is satisfiable. Thus, if $\SAT \not\in \scBPTime(N^{\polylog(N)})$, no such algorithm $\alg'$ can exist.

Proposition \ref{fullthm:weak} follows by recalling that $p=N^{O(\log^{c}(N))}$, where $c$ is chosen such that $1-\delta=\frac{c}{1+c}$; thus $\ell = 2^{\Omega(\log^c(N))} = 
2^{\Omega(\log^{\frac{c}{1+c}}(p)} =2^{\Omega(\log^{1-\delta}(p))}$. 
\end{proof}

We now turn to the description of the algorithm $\alg''$ that generates the matrix $B$. 
Our description borrows many ideas from Lund and Yannakakis \cite{lund}; accordingly, our presentation
closely follows that of \cite{lund}.

\medskip \noindent \textbf{Construction of $B$.}  \label{fullsec3.1.1}
Let $M = |A_2|$ and let $\ell>2$ be arbitrary. 
Let $\mathcal{S}_{M, \ell}=\{V_{a_2}|a_2\in A_2\}$ be a $\Delta$-useful 
set system over a set $S$, where $|S|=O(\ell^2 \cdot \log |A_2|)$, $\Delta=\Omega(|S|)$. 
Lemma \ref{fulllemma:setsystem} guarantees $\mathcal{S}_{|A_2|, \ell}$ can be output by an efficient randomized algorithm with high probability.

The matrix $B$ will have $m=|R| \times |S|$ rows, and $p= |Q_1|
\cdot |A_1| + |Q_2| \cdot |A_2|$ (which is quasipolynomial in
$N$) columns. 
We associate each of the $m$ rows of $B$ with a point $(r, s)$, 
where $r$ is a random seed to the verifier and $s$ is a point in the set $S$. Likewise, we 
associate each of the $p$ columns of $B$ with a
pair $(q_i, a_i) \in Q_i \times A_i$, for $i \in \{1, 2\}$.

For $r \in R$ and $i \in \{1,2\}$, let $q[r, i]$ be the query that the verifier asks the $i$th 
prover when using the random seed $r$. Likewise, for each $r \in R$ and each $a_1 \in A_1$, 
let $\UA[r, a_1]$ denote the unique answer such that $\cV(r,
a_1, \UA[r, a_1])=1$, 
if such an answer exists, and be undefined otherwise. By functionality of the MIP,
there exists at most one such answer. 

Now we define the following sets, following \cite{lund}. For each pair $(q_1, a_1) \in Q_1 \times A_1$, define:

\begin{equation} \label{fulleq:wq1a1} W_{q_1, a_1} = \{ \langle r, s \rangle | q_1 = q[r, 1], \UA[r, a_1]\text{ is defined, and } s \not \in V_{\UA[r, a_1]}\}.\end{equation}

Similarly, for each pair $(q_2, a_2) \in Q_2 \times A_2$, define
\begin{equation} \label{fulleq:wqiai} W_{q_2, a_2} = \{\langle r, s \rangle | q_2 = q[r, 2] \text{ and } s \in V_{a_2}\}.\end{equation}

Let $\bfw_{q_i, a_i} \in \{0, 1\}^{m}$ denote the indicator vector of $W_{q_i, a_i}$ . We define $B$ to be the $m \times p$ matrix whose columns are the $\bfw_{q_i, a_i}$ vectors over
$i=1,2$, $q_i\in Q_i, a_i\in A_i$.

Here is an equivalent way to construct $B$.    Start with the zero matrix.  
For each $r \in R$, $q_1 \in Q_1$, let the $(r,q_1)$ block of $B$ be the $|S|\times |A_1|$ submatrix corresponding to $r$ and $q_1$.
Similarly, 
for each $r \in R$, $q_2 \in Q_2$, let the $(r,q_2)$ block of $B$ be the $|S|\times |A_2|$ submatrix corresponding to $r$ and $q_2$.

Then, for each $r\in R$, do the following.
\begin{enumerate}
\item
For each $a_2\in A_2$, replace the 0-column corresponding to $a_2$ in the $(r,q[r,2])$ block by $\bfv_{a_2}$.
\item
For each $a_1\in A_1$ such that $\UA[r,a_1]$ is defined,
replace the 0-column corresponding to $a_1$ in the $(r,q[r,1])$ block by $\bar\bfv_{\UA[r,a_1]}$, 
\end{enumerate}

Note that for each $r$, there are only one value of $q \in Q_1$, namely, $q=q[r,1]$, and only one value of $q \in Q_2$, namely, $q=q[r,2]$, such that 
the $(r,q)$ block of $B$ can be nonzero.
Note also that a prover strategy corresponds to choosing, for each $i=1,2$ and each $q_i \in Q_i$, exactly one column from the block of columns corresponding to $q_i$.

\medskip \noindent \textbf{Analysis of $B$.}

\begin{proposition} \label{fullprop:completeness}
$B$ satisfies property 1 above. 
\end{proposition}
\begin{proof}
%
Assume $\phi\in \SAT$ and consider a strategy of the provers that makes the verifier accept with probability 1.
Such a strategy must exist, because the MIP has perfect completeness.
For this strategy, for $i=1,2$, let $P_i(q_i)\in A_i$ be the answer $P_i$ gives when given query $q_i$.
\begin{lemma}\label{fullcomplete}
\begin{equation}
\label{fulleq:ugh} \bfe^{(m)}=\sum_{i=1}^2 \sum_{q_i\in Q_i} \bfw_{q_i, P_i(q_i)}.\end{equation}
\end{lemma}
\begin{proof}
Consider any $r\in R$.  Because on query $r$, the responses from provers $P_1$ and $P_2$ induce acceptance,
we have ${\cal V}(r,P_1(q[r,1]),P_2(q[r,2]))=1$, and therefore
\begin{equation}\label{P1P2}
\UA(r,P_1(q[r,1]))=P_2(q[r,2]).
\end{equation}

Furthermore, the only nonzero $(r,q)$ blocks corresponding to $r\in R$ are those for $q[r,1]$ and $q[r,2]$.  It suffices to show that 
$$\bfw_{q[r,1],P_1(q[r,1])}|_r+ \bfw_{q[r,2],P_2(q[r,2])}|_r
=\bfe^{(|S|)},$$ since any other columns appearing in the sum, in the rows corresponding to $r$, are all zero.
(Here, for an $|R||S|$-vector $z$,  $z|_r$ denotes the $|S|$-vector corresponding to block $r$.)
By the definition of $B$, 
$\bfw_{q[r,2],P_2(q[r,2])}|_r=\bfv_{P_2(q[r,2])}$.  In addition,
$\bfw_{q[r,1],P_1(q[r,1])}|_r=\bar \bfv_{\UA(r,P_1(q[r,1]))}$.
But by \eqref{P1P2}, we know that $\UA(r,P_1(q[r,1]))=P_2(q[r,2])$, and therefore 
$\bfw_{q[r,1],P_1(q[r,1])}|_r+ \bfw_{q[r,2],P_2(q[r,2])}|_r=
\bar \bfv_{P_2(q[r,2])}+\bfv_{P_2(q[r,2])}=e^{|S|}$.  
\end{proof}


Thus, equation \eqref{fulleq:ugh} represents $\bfe^{(m)}$ as a linear combination (with coefficients in $\{0, 1\}$) 
of exactly $|Q_1| + |Q_2|$ columns of $B$,
proving property 1.
\end{proof}

\begin{proposition}\label{fullprop:sound}
$B$ satisfies property 2 above. 
\end{proposition}
\begin{proof}
Let $\bfx \in \mathbb{R}^{p}$ be any vector satisfying $B \bfx \in \bfB_{\Delta}(\bfe^{(m)})$. 
We claim that if $\bfx$ is sufficiently sparse, then we can transform $\bfx$ into a prover
strategy causing the verifier to accept with probability larger than $\eps_{\text{sound}}$, a
contradiction.

We say that $\bfx$ \emph{involves} $\bfw_{q_i, a_i}$ if $\bfx_{q_i, a_i} \neq 0$. 
We will write $\bfw_{q_i, a_i} \in \bfx$ to mean the same thing.

Given $\bfx$, 
and recalling that $Q_1\cap Q_2=\emptyset$,
we define a cost function $\cost: Q_1 \cup Q_2 \rightarrow \mathbb{N}$ via
$\cost(q_i) = |\{a_i \in A_i: \bfw_{q_i, a_i} \in \bfx\}|$ for $i=1,2$ and $q_i \in Q_i$. 
Similarly,
we define the cost of $r \in R$ as $\cost(r) = \cost(q[r, 1]) + \cost(q[r, 2]) $. 
We say that
$r \in R$ is \emph{good} if $\cost(r) \leq \ell$.
Let $\gamma$ denote the fraction of values $r \in R$ that are good. 
We claim that we can transform $\bfx$ into a prover strategy causing the verifier to accept with
probability at least $\gamma/\ell^2$. To establish this, we require several lemmas.

\begin{sublemma} \label{fulllemma:good}
If $r$ is good, then $\bfx$
must involve vectors $\bfw_{q[r, 1], a^*_1}$ and $\bfw_{q[r, 2], a^*_2}$
for some pair $(a^*_1, a^*_2)$ of answers,
respectively, 
such that $\cV(r, a^*_1, a^*_2)=1$.
\end{sublemma}

\begin{proof}
For $r\in R$, let $B|_{r}$ denote the $|S| \times p$ matrix obtained by restricting $B$ to rows corresponding to pairs of the form $(r, s)$ for some $s \in S$. Similarly, let $(B
\bfx)|_{r} \in \mathbb{R}^{|S|}$ denote the subvector of $B \bfx$ whose $s$th entry is $(B \bfx)_{(r, s)}$, for each $s \in S$. Note that since $\|B \bfx - \bfe^{(m)}\|^2 \leq \Delta$, it also
holds that 
\begin{equation}
\label{fulleq:restrict} \|(B \bfx)|_{r} - \bfe^{(|S|)}\|^2 \leq \Delta.\end{equation}

Since $r$ is good, the vector $(B \bfx)|_{r} \in \mathbb{R}^{|S|}$ is a linear combination of at most
$\ell$ columns of $B|_{r}$. 
To see this, observe that the definition of $B$ implies that the only non-zero columns of $B|_r$ correspond to vectors of the form $\bfw_{q[r,1], a_1}$ or $\bfw_{q[r,2], a_2}$ for some $a_1  \in
A_1$ or $a_2 \in A_2$. Since $r$ is good, $cost(q[r,1])+cost(q[r,2])\le \ell$, i.e., 
$|\{a_1\in A_1: \bfw_{q[r,1],a_1}\in \bfx\}|
+|\{a_2\in A_2: \bfw_{q[r,2],a_2}\in \bfx\}|\le \ell$, and hence
there are at most $\ell$ vectors of the form $\bfw_{q[r, 1], a_1}$ or $\bfw_{q[r, 2], a_2}$ such that $\bfx_{q[r, 1], a_1}$ or $\bfx_{q[r, 2], a_2}$ are
non-zero. Hence, $B|_r \cdot \bfx$ is a linear combination of at most $\ell$ non-zero columns of $B|_r$. Since $(B \bfx)|_{r} = B|_r \cdot \bfx$, it follows that $(B \bfx)|_{r}$ is a linear
combination of at most $\ell$ columns of $B|_r$ as claimed.

Moreover, 
each column of $B|_{r}$ is equal to either $\bfv_{a_2}$ (see equation \eqref{fulleq:wqiai}) or $\bar\bfv_{a_2}$ (see equation \eqref{fulleq:wq1a1}) for some $a_2\in A_2$ 
(recall that $\bfv_{a_2}$ denotes the indicator
vector of set $V_{a_2}$ in the set system $\calB_{M, \ell}$).
Since
$\mathcal{B}_{M, \ell}$ is $\Delta$-useful, the definition of $\Delta$-useful set systems (Definition \ref{fulldef:setsystem}) 
and equation \eqref{fulleq:restrict} together
imply that there is some $a_2^*$ such that both $\bfv_{a_2^*}$ and 
$z:=\bar\bfv_{a_2^*}$ 
are involved in the linear combination representing 
$(B \bfx)|_{r}$ as
described in the previous paragraph. 
(Because $\ell>2$, 
the definition of $\Delta$-useful implies that there is no $a\in A_2$ with $a\ne a^*_2$
such that $z= \bar\bfv_{a}$  and that there is no $a\in A_2$ such that 
$z= \bfv_{a}$.)
Now equation \eqref{fulleq:wq1a1} implies that there
exists a (possibly not unique) 
$a^*_1\in A_1$ such that $\UA[r,a^*_1]=a^*_2$.  

%
%

In summary, we have established that
$\bfx$ involves $\bfw_{q[r, 1], a^*_1}$
and $\bfw_{q[r, 2], a^*_2}$ for some $(a^*_1, a^*_2)$ such that $\cV(r, a^*_1, a^*_2)=1$, proving the sublemma.
\end{proof}

%
%

\begin{sublemma} \label{fulllemma:strategy} There exists a pair of provers such that the verifier accepts the input $\phi$ 
with probability at least $\gamma/\ell^2$. 
\end{sublemma}
\begin{proof}
For each prover $i \in \{1, 2\}$, define $\ell$ strategies $\cP_{i, 1}, \dots, \cP_{i, \ell}$ as follows: 
For each query $q_i \in Q_i$, order the elements of the set $\{a_i \in A_i : \bfw_{q_i,a_i} \in \bfx\}$  (which can have size $cost(q_i)$ exceeding $\ell$)
arbitrarily and, for $1\le j\le \ell$, make the $j$th strategy's answer to $q_i$ be the $j$th element of the 
set, if such an element exists; otherwise, 
assign an arbitrary answer. 

If $r$ is good, then $|\{a_i\in A_i: \bfw_{q[r,i],a_i}\in \bfx\}|\le \ell$ for all $i$, because 
$\cost(q[r, 1])$ and $\cost(q[r, 2])$ are both at most $\ell$.
Sublemma \ref{fulllemma:good} implies that 
$\bfx$
must involve vectors $\bfw_{q[r, 1], a^*_1}$ and  $\bfw_{q[r, 2], a^*_2}$ that correspond to answers $(a^*_1, a^*_2)$ that make the verifier accept when using
random choice $r$. 
It follows, since $\gamma$ is the fraction of $r\in R$ that are good, that if the two provers execute the following (randomized) strategy, the verifier will accept with probability at least $\gamma/\ell^2$: on query $q_i$, the $i$th prover chooses $j \in \{1, \dots, \ell\}$ at random, and replies according to strategy $\cP_{i, j}$.
By averaging, this implies that there is some fixed pair $(j_1, j_2) \in \{1, \dots, \ell\}^2$ such that $(\cP_{1, j_1}, \cP_{2, j_2})$ makes the verifier accept with probability at least $\gamma/\ell^2$.   
\end{proof}

\begin{sublemma} \label{fulllemma:claim2}
$\|\bfx\|_0 \geq (1-\gamma) \frac{\ell}{2} \left(|Q_1| + |Q_2|\right)$.
\end{sublemma}
\begin{proof}
Note that 

\begin{equation}
\label{fulleq:eq1} \sum_{r \in R} \cost(r) \geq \sum_{r \in R \colon r \text{ is not good}} \cost(r) \geq \left [ (1-\gamma) |R|\right ] \cdot \ell.\end{equation}
The second inequality follows from $$|\{r\in R \colon r\mbox{ is
not good}\}|=(1-\gamma)|R|$$ and $cost(r)>\ell$ if $r$ is not good.

We can also write:

\begin{equation} \label{fulleq:eq2} \sum_{r \in R} \cost(r) =  \sum_{r \in R} \sum_{i=1}^2 \cost(q[r, i])  =   \sum_{i=1}^2\left [ \sum_{r \in R} \cost(q[r, i]) \right ]. \end{equation}

By uniformity (see Definition \ref{fulldef:canonical}), for each $i \in \{1, 2\}$, $\sum_{r \in R} \cost(q[r, i]) = [ \sum_{q_i\in Q_i} cost(q_i)]\frac{|R|}{|Q_i|}$. Hence, the right-hand side
of equation \eqref{fulleq:eq2} equals
\begin{equation} \label{fulleq:expression}
\sum_{i=1}^2 \sum_{q_i \in Q_i} \cost(q_i) \cdot \frac{|R|}{|Q_i|}.
\end{equation}
Since $|Q_1| = |Q_2|$ by equality of question-space sizes, expression \eqref{fulleq:expression} implies that
\begin{equation} \label{fulleq:final}  \sum_{r \in R} \cost(r) =\frac{|R|}{|Q_1|} \cdot \sum_{i=1}^2 \sum_{q_i \in Q_i} \cost(q_i)  = \frac{|R|}{|Q_1|} \|\bfx\|_0,\end{equation}
 where the final equality holds because $\sum_{i=1}^2 \sum_{q_i\in Q_i}\cost(q_i)=\|\bfx\|_0$.

Combining inequality \eqref{fulleq:eq1} and equation \eqref{fulleq:final}, we conclude that 
$\|\bfx\|_0\frac {|R|}{|Q_1|} \ge (1-\gamma) |R|\ell$, from which we conclude that
$\|\bfx\|_0 \geq (1-\gamma) \frac{\ell}{2} \left(|Q_1| + |Q_2|\right)$ as claimed.
\end{proof}

Sublemma \ref{fulllemma:strategy} implies that if $\phi \not\in
\SAT$, then $\gamma/\ell^2 \le \eps_{\text{sound}}$
and hence $1-\gamma\ge 1-\ell^2 \eps_{\text{sound}}$. 
Hence, using Sublemma \ref{fulllemma:claim2}, 
$$\|\bfx\|_0 \geq (1-\gamma) \frac{\ell}{2} \left(|Q_1| + |Q_2|\right) \geq \left(1-\ell^2 \cdot \eps_{\text{sound}}\right) \frac{\ell}{2} \left(|Q_1| + |Q_2|\right).$$
This completes the proof of Proposition \ref{fullprop:sound}.
\end{proof}

\subsection{Proof of Theorem \ref{fullCSthmstrong}}
\label{fullsec:projection}
Let $\PCP_{1, \eps_{\text{sound}}}[b, k]_{|\Sigma|}$ denote the class of languages that have a PCP verifier with perfect completeness,
soundness $\eps_{\text{sound}}$, $b$ bits of randomness used by the verifier, and $k$ queries to a proof over alphabet $\Sigma$.
The \emph{sliding scale conjecture} of Bellare et al. \cite{bellare} 
postulates that for every $0 \leq \delta \leq 1$,  $\SAT \in \PCP_{1, \eps}[O(\log(N)), 2]_{\poly(1/\eps)}$, where $\eps = 2^{-\Omega(\log^{1-\delta}(N))}$. In words, the sliding scale conjecture asserts that $\SAT$ has a polynomial-length PCP in which the verifier makes
two queries to the proof (note that each query reads an entire
symbol from the alphabet $\Sigma$, and hence corresponds to
approximately $\log |\Sigma|$ consecutive bits of the proof), and in which the soundness error falls exponentially with the number of bits of the proof that the verifier accesses. 

Note that any $k$-query PCP can be transformed into a $k$-prover MIP by posing each of the PCP verifier's $k$ queries to a different prover. Thus, the sliding scale
conjecture (at $\delta=0$) implies
the existence of a 2-prover MIP for $\SAT$ in which $|R|, |Q_1|, |A_1|,
|Q_2|, |A_2|$ are all at most $\poly(N)$, and whose soundness
error is polynomially small in $N$. 

In \cite{projectiongames}, Moshkovitz formulates a slight strengthening of the sliding scale conjecture
that she
calls the \emph{Projection Games Conjecture} (PGC). Here, the term \emph{projection games} refers to a certain 
class of two-player one-round games defined over bipartite graphs. 
Projection games are equivalent, in a formal sense, to two-prover MIP's that satisfy functionality. This equivalence is standard; we provide the details in the appendix for completeness.

Just like the sliding scale conjecture, the PGC postulates the
existence of a 2-prover MIP for $\SAT$ in which $|R|, |Q_1|, |A_1|, |Q_2|, |A_2|$ are all at most $\poly(N)$, and in which the soundness
error is inversely polynomially small in $N$. But the PGC additionally postulates that there is such an MIP that satisfies the functionality property used in the proof of Proposition \ref{fullthm:weak} (see Definition \ref{fulldef:canonical}). 
Formally, the PGC states the following.\footnote{Technically, Moshkovitz's PGC is slightly stronger than our Conjecture \ref{fullconjecture}
in that $|R|$, $|Q_1|$, and $|Q_2|$ are all required to be bounded by $N^{1+o(1)} \poly(1/\eps)$, rather than $\poly(N, 1/\eps)$. She states the weaker version in a footnote, and the weaker version suffices for our purposes.}

\begin{conjecture}[Reformulation of the projection games conjecture (PGC) \cite{projection games}]
\label{fullconj16}
There exists $c>0$ such that for all $N$ and all $\eps \geq 1/N^c$, there is a
2-prover MIP for \SAT\ instances of size $N$,  with perfect completeness and soundness error at most $\eps$, which satisfies the following two properties:
\begin{itemize}
\item $|R|, |Q_1|, |A_1|, |Q_2|, |A_2|$ are all at most $\poly(N, 1/\eps)$.
\item The MIP satisfies functionality.
\end{itemize}
\end{conjecture}

We will need a slight strengthening of the PGC that requires that $\SAT$ be efficiently reduced to a 2-prover MIP
that satisfies uniformity in addition to functionality. 
We note that in the language of projection games, the uniformity condition is equivalent to requiring the bipartite graph underlying
the projection game be biregular. Moreover, the most efficient known reductions from $\SAT$ to projection games do produce biregular graphs
(and hence 2-prover MIP's satisfying uniformity) \cite{steurer}.

\begin{conjecture}[Biregular PGC]
\label{fullconjecture2}
Conjecture \ref{fullconj16} holds in which the MIP satisfies uniformity as well.
\end{conjecture}

Proposition \ref{fullthm:strongformal} below describes a strengthening of Proposition \ref{fullthm:weak} that holds under the Biregular PGC. Theorem \ref{fullCSthmstrong} then follows by the argument of Section \ref{fullsec:proofoverview}, using Proposition \ref{fullthm:transform} in place of Proposition \ref{fullthm:strongformal}.

\begin{proposition}
\label{fullthm:strongformal}
Assuming Conjecture \ref{fullconjecture2}, 
the following holds for some
pair of values
$\sigma=\sigma(p), \Delta=\Delta(p)$, both in $p^{\Omega(1)}$.

If $\SAT \not \in \scBPP$, then
there are a  polynomial $m=m(p)$ and $k=k(p)$
such that no $\scBPP$ algorithm distinguishes the following two cases, given 
an  $m \times p$ Boolean matrix $B$:

\begin{enumerate}
\item There is an $\bfx \in \{0, 1\}^p$ such that $B\bfx=\bfe^{(m)}$ and $\|\bfx\|_0 \leq k$.
\item For all $\bfx \in  \mathbb{R}^p$ such that $B\bfx \in
\bfB_{\Delta}(\bfe^{(m)})$, $\|\bfx\|_0 \geq k \cdot
\sigma$.
\end{enumerate}
\end{proposition}

\begin{proof}
Conjecture \ref{fullconjecture2} implies the existence of a
perfectly complete two-prover MIP for $\SAT$ with soundness
error $\eps_{\text{sound}} = 1/N^{\Omega(1)}$ that satisfies
functionality and uniformity, with $|R|, |A_1|, |A_2|, |Q_1|, |Q_2|$ all bounded by $\poly(N)$. 
Moreover, this MIP can be transformed into an equivalent MIP that satisfies both equality of question space sizes and disjointness of answer spaces, while keeping $|R|, |A_1|, |A_2|, |Q_1|, |Q_2|$ bounded by $\poly(N)$. 
Disjointness of answer spaces can be ensured by simple padding (only one bit of padding is required). Equality of question space sizes can be ensured via a technique of Lund and Yannakakis \cite{lund}, with at most a quadratic blowup in the size of $Q_1$ and $Q_2$. Namely, 
the verifier generates 
two independent queries $(q_1, q_2)$ and $(q'_1, q'_2)$ in $Q_1 \times Q_2$. 
Then, the verifier asks the first prover the 
query $(q_1, q'_2)$ and the second prover the query $(q_2, q'_1)$. The provers ignore the second component, and answer the 
queries by answering the first component of the query according to the 
original MIP. The verifier accepts if 
and only if the verifier in the original MIP would have accepted the answers to $(q_1, q_2)$. 

Thus, we have established that Conjecture \ref{fullconjecture2} implies the existence of a perfectly complete canonical MIP for $\SAT$ with soundness error $\eps_{\text{sound}} = 1/N^{\Omega(1)}$ and for which $|R|, |A_1|, |A_2|, |Q_1|, |Q_2|$ are all bounded by $\poly(N)$. 
The remainder of the proof is now identical to that of Proposition \ref{fullthm:weak}; we provide the details for completeness. 

The proof of Proposition \ref{fullthm:weak} specified an algorithm $\alg''$ that output a Boolean $m \times p$ matrix $B$
of size polynomial in the parameters $|R|,  |A_1|, |A_2|,
|Q_1|, |Q_2|$, which are themselves polynomial in $N$
(so that $mp$ is polynomial in $N$), such that:
\begin{itemize}
\item[Property 1:] If $\phi \in \SAT$, then there is a vector $\bfx^* \in \{0, 1\}^p$ such that $B \bfx^*=\bfe^{(m)}$ and 
$\|\bfx^*\|_0 \leq |Q_1| + |Q_2|$. 
\item[Property 2:] Abusing notation, let $\ell:=\ell(N)$. If $\phi \not\in \SAT$, then any $\bfx \in \mathbb{R}^p$ such that $B\bfx \in \bfB_{\Delta}(\bfe^{(m)})$
satisfies $\|\bfx\|_0 \geq (1-\eps_{\text{sound}} \cdot \ell^2) \cdot \frac{\ell}{2} \cdot \left( |Q_1| + |Q_2|\right)$,
for some $\Delta = \Omega\left(\ell^2 \cdot \log \left( |A_1| + |A_2| \right)\right)$.
\end{itemize}

Setting $\ell = (10\cdot \eps_{\text{sound}})^{-1/2}$ and $k=|Q_1| + |Q_2|$, properties 1 and 2 imply that
\begin{itemize}
\item If $\phi \in \SAT$, then there is a vector $\bfx^* \in \{0, 1\}^p$ such that $B \bfx^*=\bfe^{(m)}$ and 
$\|\bfx^*\|_0 \leq k$. 
\item If $\phi \not\in \SAT$, then any $\bfx \in \mathbb{R}^p$ such that $B\bfx \in \bfB_{\Delta}(\bfe^{(m)})$
satisfies $\|\bfx\|_0 \geq \sigma \cdot k$, where $\Delta(p) = \Omega\left(\ell^2 \cdot \log \left( |A_1| + |A_2| \right)\right) \geq \Omega(\ell)=p^{\Omega(1)}$,
and $\sigma(p) = (1-\eps_{\text{sound}} \cdot \ell^2) \cdot
\frac{\ell}{2} \geq \Omega(\ell) = p^{\Omega(1)}$.   (Recall
that $p$ is polynomial in $N$.)
\end{itemize}

Suppose there is a randomized algorithm $\alg'$ that runs in
$\scBPTime(\poly(m, p)) = \scBPTime(\poly(N))$ 
and distinguishes between 
the above two cases. By running $\alg'$ on the matrix $B$ output by $\alg''$, we determine with high probability whether 
$\phi$ is satisfiable. Thus, if $\SAT \not\in \scBPTime(\poly(N))$, no such algorithm $\alg'$ can exist.
\end{proof}

\subsection{Proof of Proposition \ref{fullthm:transform}}\label{fullprop7}
\begin{proof}
Let $g(\cdot)$ be any function and let $h(m, p) = p^{C_1} \cdot m^{1-C_2}$ for some positive constants $C_1, C_2$. 
Let $\alg$ be any algorithm for \gSR\ that runs in randomized time $T(m, p)$. Then on any $m \times p$ input matrix $B$
such that there exists a vector $\bfx^*$ with $\|\bfx^*\|_0 \le k$, $\alg$ outputs a vector $\bfx$ with $\|\bfx\|_0 \leq g(p) \cdot k$
satisfying $\|B\bfx-\bfe^{(m)}\|^2 \leq  p^{C_1} \cdot m^{1-C_2}$. 

We show how to transform $\alg$ into an algorithm $\alg'$ for the $(g, 1)$-\textsc{Sparse Regression} problem, such that $\alg'$ runs in time $T(p^{C_1/C_2} \cdot m^{1/C_2}, p)$. 

Let $r= \left \lceil \left(p^{C_1} \cdot
m^{1-C_2}\right)^{1/C_2}\right \rceil$.  Let $B'$ denote the $(mr)
\times p$ matrix obtained by stacking $r$ copies of $B$ atop each other. The algorithm $\alg'$ simply runs $\alg$ on $B'$, to obtain an approximate solution $\bfx$ to the system of linear equations given by $B' \bfx = \bfe^{(rm)}$, and outputs $\bfx$. 
Notice that the running time of $\alg'$ is indeed $O(T(p^{C_1/C_2} \cdot m^{1/C_2}, p) = O\left(T\left(\poly\left(m, p\right), p\right)\right)$.

\medskip 
\noindent \textbf{Analysis of $\bfx$}: Since $B\bfx^* = \mathbf{e}^{(m)}$, it also holds that $B'\bfx^* = \bfe^{(rm)}$. That is, the system $B'\bfx = \mathbf{e}^{(rm)}$ has an exact solution of sparsity $k$. Hence, $\alg$
must output a vector $\bfx$ with $\|\bfx\|_0 \leq g(p) \cdot k$, satisfying
$\|B'\bfx-\bfe^{(rm)}\|^2 \leq  p^{C_1} \cdot (rm)^{1-C_2}$. 
Notice that:
$$\|B'\bfx-\bfe^{(rm)}\|^2  = r \cdot \|B\bfx-\bfe^{(m)}\|^2.$$ Thus,
$$\|B\bfx-\bfe^{(m)}\|^2 =(1/r) \cdot \left [
\|B'\bfx-\bfe^{(rm)}\|^2\right ]   \leq  (1/r) \left [
p^{C_1} \cdot (rm)^{1-C_2}\right ] =  p^{C_1} \cdot m^{1-C_2}
\cdot r^{-C_2} \le 1.$$
This completes the proof.  
\end{proof}

\section{The Statistical Problem}\label{fullnoisy}
\subsection{Motivation and Prior Work}\label{fullnoisymotivation}
In this section we motivate the
\nSR\ problem introduced in Section \ref{fullsec:intro}. 
The problem \nSR\ is a sparse variant of a 
less challenging problem known as \textsc{Noisy Regression} that is popular in the 
statistics literature. In \textsc{Noisy Regression}, instead of being given a fixed 
matrix $B$ for which we seek a sparse vector $\bfx$ such that $B\bfx
\approx \bfe$, we assume corrupting random noise injected into the problem.  
Specifically, we assume that there are (1) a known $m \times p$
matrix $X$ (which plays the role of $B$) and (2) an unknown real $p$-dimensional vector $\bftheta$.  An
$m$-dimensional vector $\bfeps$ of random noise is generated, with the
$\epsilon_i$'s being i.i.d. $N(0,1)$ random variables.  The vector
$\bfy=X\bftheta+\bfeps$ (but not $\bfeps$ itself) is then revealed to the algorithm (along
with $X$, which it already knew).  
On an instance specified by $(X,\bftheta)$ with $\bfy$ revealed to the algorithm, an algorithm will produce $\hat \bftheta=\hat
\bftheta(X,\bfy)$.  We define $\hat{\bfy}=X\hat \bftheta$.  The instance's {\em prediction loss} is defined to be $\|\hat{\bfy}-E[\bfy]\|^2=\|X(\hat \bftheta-\bftheta)\|^2$.  
(By ``$E[\bfy]$'' here, since $\bfy=X\bftheta+\bfeps$, we mean the {\em vector} $X\bftheta$.)
Its expected value $E[\|X(\hat \bftheta-\bftheta)\|^2]$, over the random vector $\bfeps$, 
is defined to be the {\em risk}.  

The goal is to make the risk as small as possible.  However, as $\bftheta$ is unknown, it is hard to measure the performance of an algorithm at a 
single $\bftheta$.  After all, an algorithm which luckily guesses $\bftheta$ and then sets $\hat \bftheta=\bftheta$ will have zero risk for that $\bftheta$.  For
this reason, one usually seeks an algorithm that minimizes the supremum of the risk over all $\bftheta$.  In other words, this {\em minimax estimator}
finds $\hat \bftheta$ to minimize the supremum over $\bftheta$ of $E[\|X(\hat \bftheta-\bftheta)\|^2]$.
\textsc{Noisy Regression} refers to the  problem of minimizing this last supremum.

Note here that we are using random noise 
whose components have variance 1, regardless of the
magnitudes of the entries of $X$. For our hard instances, the entries of $X$ are 0-1.

It is well-known that ordinary least squares is a minimax
estimator: least squares achieves risk $p$, \emph{independently}
of $m$ \cite[bottom of page 1950]{FG}.
Moreover, no estimator, polynomial-time or not, achieves smaller maximum risk.
(It's not even {\em a priori} obvious that the risk can be 
made independent of $m$.)

The problem \nSR\ defined in Section \ref{fullsec:intro} is
a 
``sparse version'' of {\sc Noisy Regression} which is
identical to {\sc Noisy Regression} except that there
is a positive integer $k$, known to the algorithm, such that
$||\bftheta||_0\le k$.
The goal of the algorithm is to find an estimate
$\hat \bftheta$ of $\bftheta$ for which both the supremum over $\bftheta$ of the sparsity
$||\hat \bftheta||_0$ and the supremum over $\bftheta$ of the risk 
$E[\|X(\hat \bftheta-\bftheta)\|^2]$ are as small as possible.

As was the case for \gSR, there is a naive exponential-time algorithm for {\sc Sparse
Noisy Regression}.  Namely, try all subsets $S$ of $\{1,2,...,p\}$ of
size $k$.  For each $S$, 
use ordinary least squares to find the $\hat \bftheta_j$'s, except require that $\hat \bftheta_j=0$ for all $j\not
\in S$.  An upper bound on the risk of this algorithm is known:
\begin{theorem} \cite[page 1962]{FG}
For any $\bftheta$, the risk of the naive exponential-time algorithm is bounded by
$4k\ln p$.  
\end{theorem}

The question we are interested in is, how small can a {\em polynomial-time} algorithm make
the risk, if it is allowed to ``cheat'' somewhat on the
sparsity?  More formally, if a polynomial-time algorithm is
allowed to output a vector $\hat \bftheta$ with $\|\hat \bftheta\|_0\le
k\cdot g(k)$ (where $g(k)\ge 1$ is an
arbitrary function), how small can its risk be?

Theorem \ref{fullstatthm} (respectively, Theorem
\ref{fullstatthmstrong}) from Section \ref{fullsec:intro} assert that the risk cannot be bounded above by $p^{C_1} \cdot m^{1-C_2}$ for
\emph{any} positive constants $C_1, C_2$, even if one is allowed to cheat by a nearly polynomial (respectively, polynomial) factor on the sparsity of the returned vector. This is in stark contrast
to {\sc Noisy Regression}, where there is a polynomial-time algorithm (least-squares regression)
that achieves risk $p$, independently of $m$.

Now we can describe the result of Zhang, Wainwright, and Jordan
\cite{ZWJ}.  There is a standard (non-greedy)
algorithm known as LASSO which is known to achieve risk bounded
by $O((1/\gamma^2(B)) (k\log p))$, where $\gamma(B)$ is $B$'s
``restricted eigenvalue constant'' \cite{ZWJ}.  Zhang,
Wainwright, and Jordan prove that, unless NP$\subseteq$P/poly,
for any $\delta>0$
any polynomial-time algorithm has risk
$\Omega((1/\gamma^2(B))(k^{1-\delta}\log p))$, for some matrices
$B$ (for which $\gamma(B)$ can be made arbitrarily
small).  To contrast their results with ours, \cite{ZWJ} requires
a different complexity assumption and proves a different lower
bound (involving $\gamma(B)$, which ours does not), but most
importantly, their bound only applies to algorithms that return
linear combinations of $k$ columns, whereas ours allows linear
combinations of $2^{O(\log^{(1-\delta)}p)} \cdot k$ columns.

\subsection{Proof of Theorems \ref{fullstatthm} and \ref{fullstatthmstrong}}
%
%
%
%
%
%

\begin{proposition}  \label{fullthm:stattocs}
For every pair $g(\cdot), h(\cdot, \cdot)$ of polynomials, if
there is an algorithm for $(g, h/2)$-\textsc{Noisy Sparse Regression} that runs in time $T(m, p)$,
then there is an algorithm for \gSR\ that runs in time $O(T(m, p))$.
\end{proposition}

Combining Proposition \ref{fullthm:stattocs} with Theorem \ref{fullCSthm} proves Theorem \ref{fullstatthm}
and combining Proposition \ref{fullthm:stattocs} with Theorem \ref{fullCSthmstrong} proves Theorem \ref{fullstatthmstrong}.

\begin{proof}
Let $\alg$ be an algorithm for $(g, h/2)$-\textsc{Noisy Sparse Regression}. The following algorithm $\alg'$ 
solves \gSR, with failure probability at most $\delta$.

$\alg'$ does the following $\lceil
\log(1/\delta)\rceil $ times: it generates an $m$-dimensional
vector $\bfeps$ whose components are (approximately) i.i.d.
$N(0,1)$ random variables and sets $\bfy=\bfe^{(m)}+\bfeps$. (We will assume that the algorithm generates $N(0,1)$ random variables exactly.)
It then runs $\alg$ on input $(B, k, \bfy)$ to obtain a vector
$\bfx$ of sparsity at most $k\cdot g(p)$, and checks whether $\|B \bfx - \bfe^{(m)}\|^2 \leq
h(m, p)$. If so, it halts and outputs $\bfx$. If $\alg'$ has not
halted after $\lceil \lg (1/\delta)\rceil$ iterations of the above procedure, $\alg'$ outputs a special failure symbol $\perp$.

We must argue, given an $m \times p$ input matrix $B$ such that there is a $p$-dimensional
vector $\bfx^*$ satisfying $\|\bfx^*\|_0 \le k$ and $B
\bfx^*=\bfe^{(m)}$, that $\alg'$ outputs
a vector $\bfx$ of sparsity at most $k \cdot g(p)$ satisfying
$\|B \bfx - \bfe^{(m)}\|^2 \leq h(m, p)$, with probability at least $1-\delta$.
Clearly if $\alg'$ does not output $\perp$, then the vector
$\bfx$ that it outputs satisfies the required conditions. We argue
that the probability $\alg'$ outputs $\perp$ is at most
$1/2^{\lceil\log(1/\delta)\rceil } \le \delta$.

Since $\bfe^{(m)}=B\bfx^*$, each invocation of $\alg$ must
return a vector $\bfx$ which must satisfy
(a) $\|\bfx\|_0\le k\cdot g(p)$ and (b) $E[\|B(\bfx-\bfx^*)\|^2] = E[\|B\bfx-\bfe^{(m)}\|^2] \le h(m, p)/2$, where the expectation is taken over the internal
randomness of $\alg$ and the random choice of $\bfeps$. By Markov's inequality, the probability that $\|B\bfx-\bfe^{(m)}\|^2 >h(m, p)$ is at most $1/2$. Hence,
the probability that $\alg$ returns a vector $\bfx$ such that
$\|B(\bfx-\bfx^*)\|^2 >h(m, p)$ during $\lceil
\log(1/\delta)\rceil $ consecutive invocations is bounded above
by $1/2^{\lceil \log(1/\delta)\rceil}$. This completes the proof of the theorem.
\end{proof}

\section{Acknowledgments}
The authors thank 
Dana Moshkowitz and Subhash Khot for their help with multiple prover proof systems.


\bibliographystyle{abbrv}
\bibliography{srrefs}
\appendix

\section{Equivalence Between Projection Games and MIP's Satisfying Functionality}
\label{fullapp:projectiongames}
Formally, a projection game consists of (a) a bipartite multigraph $G
= (Q_1, Q_2, R)$, (b) finite alphabets $A_1,A_2$, and
(c) projection constraints, specified via a partial function $\pi_r\colon A_1\rightarrow A_2$
for every edge $r \in R$.\footnote{In \cite{projectiongames}, Moshkovitz
defines $\pi_r$ to be a total function. We allow $\pi_r$ to be a partial function in accordance with
standard definitions \cite{hochbaum}, and to ensure a precise equivalence between projection games
and 2-prover MIP's satisfying functionality.}
 The goal is to
find assignments $P_1 \colon Q_1 \rightarrow A_1$ and
$P_2 \colon Q_2 \rightarrow A_2$ 
that satisfy as many edges as possible, where $(P_1,P_2)$
is said to satisfy an edge $r=(q_1, q_2)$ if
$\pi_r(P_1(q_1))=P_2(q_2)$. The \emph{size of the projection
game} is defined to be $n=|Q_1|+|Q_2|+|R|$, and the \emph{size of the
alphabet of the projection game} is defined to be $\max\{|A_1|, |A_2|\}$.
The {\em acceptance probability} of the projection game is the maximum, over $P_1:Q_1\rightarrow A_1$ and $P_2: Q_2\rightarrow A_2$, of the fraction of 
edges that are satisfied by the pair $(P_1,P_2)$ of provers.

Equivalence between projection games and 2-prover MIP's satisfying functionality should be clear, except possibly
for the $\pi_r$'s,
in light of the letters used in the definition of a projection game.  
In an MIP satisfying functionality, for each $r$ and $a_1$, there is at most one $a_2$ such that $\cV(r,a_1,a_2)=1$. 
Given $r$ and $a_1$, if $a_2$ exists with $\cV(r,a_1,a_2)=1$, then $a_2$ is unique; define $\pi_r(a_1)=a_2$ in this case. If no such $a_2$ exists, leave $\pi_r(a_1)$ undefined. 
Conversely, if $\pi_r(a_1)=a_2$, 
let $\cV(r,a_1,a_2)=1$ and let $\cV(r,a_1,a)=0$ for all other $a\in A_2$, and if $\pi_r(a_1)$ is undefined, let $\cV(r,a_1,a)=0$
for all $a\in A_2$.  In this way we have defined $\cV(r,a_1,a_2)$ for all $r\in R$, $a_1\in A_1$, and $a_2\in A_2$.

In \cite{projection games}, Moshkovitz states the Projection Games Conjecture as follows.\footnote{Technically, Moshkovitz's PGC is slightly stronger than our Conjecture \ref{fullconjecture}
in that size of the projection game is required to be $n=N^{1+o(1)} \poly(1/\eps)$, rather than $\poly(N, 1/\eps)$. She states the weaker version in a footnote, and the weaker version suffices for our purposes.}

\begin{conjecture}[Projection Games Conjecture (PGC) \cite{projection games}]
\label{fullconjecture}
There exists a $c>0$ such that for all $N$ and $\eps \geq 1/n^c$, $\SAT$
on inputs of size $N$ can be efficiently reduced to projection
games of size $n=\poly(N, 1/\eps)$ over an alphabet of size
$\poly(1/\eps)$, such that a $\SAT$ input which is satisfiable maps to a projection game of acceptance probability 1
and a $\SAT$ input which is not satisfiable maps to a projection game of acceptance probability at most $\eps$.  
\end{conjecture}

The equivalence of projection games and 2-prover MIP's with functionality implies that the PGC is equivalent to Conjecture \ref{fullconj16} from Section \ref{fullsec:projection}.

\section{A Bad Example for Natarajan's Algorithm}
First we give the problem solved by Natarajan's algorithm.  We are given a
collection $\cal C$ of vectors (the columns of a matrix $B$) and a target
vector $\bfb$.  The algorithm successively picks columns from $\cal
C$ in an attempt to represent $\bfb$ approximately as a linear
combination of the selected columns.
The algorithm
guarantees that the number of columns selected by the algorithm
is at most $\lceil Opt(\eps/2)\left [18||\bfBB^+||^2\ln
(||\bfb||/\eps)\right ]\rceil$ and that there is a linear combination of the selected variables within $L_2$ distance
at most $\eps$ of $\bfb$.  Here $Opt(\eps/2)$ denotes the minimum
sparsity of a linear combination with $L_2$ error at most $\eps/2$ and
$\bfBB^+$ is the pseudoinverse of the matrix obtained from
$B$ by normalizing each column.  Our goal is to show that the
$||\bfBB^+||^2$ factor is necessary.

The algorithm is as follows.  Start with $S=\emptyset$, $V=\cal
C$, and
$\bfb'=\bfb$.  At the end of a generic iteration of the algorithm, we know that the
vectors in $S$ are orthogonal, that each is orthogonal to
$\bfb'$ and to all $\bfv\in V-S$.
In the $i$th iteration, the algorithm chooses the $\bfv\not \in S$
which has maximum value of $(\bfb'\cdot \bfv)/||\bfv||$ (in other words,
has smallest angle with $\bfb'$).
That vector $\bfv^*$ is added to $S$.  Now we make
$\bfb'$ and the vectors in $V-S$ orthogonal to $\bfv^*$ by transforming
those vectors,
by setting $\bfb':=\bfb'-(\bfb'\cdot \bfv^*)/(\bfv^*\cdot \bfv^*)$ and
$\bfu:=\bfu-(\bfu\cdot \bfv^*)/(\bfv^*\cdot \bfv^*)$ for each $\bfu\in V-S$.
This vector-selection phase continues as long as $||\bfb'||>\eps$.

The span of the vectors in $S$ is
the same as that of the ``corresponding'' vectors in $\cal C$,
because the linear operations made during the execution of
the algorithm do not change the spans.    When we are finished
selecting vectors, we write $\bfb-\bfb'$ as a linear combination
of the vectors in $S$ or, equivalently, of the vectors in $\cal
C$ corresponding to the vectors in $S$.

The example is simple.  We simply have to show that the
vector-selection phase selects a large number of vectors, relative to
$Opt(\eps/2)$.
The vector $\bfb=\bfe^{(m)}$.  There
are $m+2$ vectors.  The first $m/2$, $\bfv_1,\bfv_2,...,\bfv_{m/2}$, are binary and have two 1's
each, $\bfv_i$ having ones in positions $2i-1$ and $2i$.
There are two additional vectors, which we call $\bfv^+$ and $\bfv^-$.
We start by setting $\bfv^+=\bfgamma:=(1,-1,1,-1,1,-1,...,1,-1)^T\in \RR^m$
and $\bfv^-=-\bfgamma$,
and then we add $\delta\bfe^{(m)}$ to $\bfv^+$ and $\bfv^-$, where
$\delta=1/\sqrt m$.
Set ${\bfv^+}'=\bfv^+$ and ${\bfv^-}'=\bfv^-$.

It is clear that $(\bfv^++\bfv^-)/(2\delta)=\bfb$, so that there is an
exact way to write $\bfb$ as a linear combination of two vectors.
Let $\eps=0.5\sqrt 2$.
Therefore
$Opt(\eps/2)=2$.
We will show that the algorithm may add
the vectors in the order $\bfv_1,\bfv_2,...,\bfv_{m/2}$, not
terminating until it's added $m/2$ vectors.  Since
$\ln(||\bfb||/\eps)=\ln((\sqrt m)/(1/2))=\ln (2\sqrt m)\le
1+0.5\cdot \ln m$, the number of vectors produced by the algorithm
should be at most $\lceil 2\left [18\cdot||\bfBB^+||^2(1+0.5 \ln
m)\right ]\rceil$; we infer that $||\bfBB^+||^2$ must be $\Omega(m/\log
m)$ and that the $||\bfBB^+||^2$ factor is necessary in the
analysis.

Let $\bfp_j$ denote the $m$-vector which starts with $m-j$ 0's and ends
with $j$ 1's.
We will prove inductively that at the beginning of the $h$th
iteration, $h=0,1,2,...,m/2-1$, $S=\{\bfv_1,\bfv_2,...,\bfv_h\}$;
$V-S$ consists of $\bfv_{h+1},...,\bfv_{m/2}$, together with
${\bfv^+}'$
and ${\bfv^-}'$ (both of which change over time), where
${\bfv^+}'=\bfgamma+\delta\cdot \bfp_{m-2h}$ and
${\bfv^-}'=-\bfgamma+\delta\cdot \bfp_{m-2h}$; and that
the vector $\bfb'$ (which changes over time) equals $\bfp_{m-2h}$.

It is easy to see that the invariant is true at the beginning of
the 0th iteration. 

Now we start the inductive proof.  Choose any $h\ge 0$ and
assume that the invariant holds.  To decide which vector to add,
we compare $(\bfb'\cdot \bfv_i)/||\bfv_i||=2/\sqrt 2=\sqrt 2$ (for any
$i\ge h+1$)
to $(\bfb'\cdot {\bfv^+}')/||{\bfv^+}'||
=(\bfb'\cdot
{\bfv^-}')/||{\bfv^-}'||=((m-2h)\delta)/\sqrt{(m-2h)(1+\delta^2)}=\sqrt{m-2h}\left
[ \frac \delta {\sqrt{1+\delta^2}}\right ]$.
Because $\delta<\sqrt 2/\sqrt m$, the former is larger, and hence
we can assume that the algorithm adds $\bfv_{h+1}$ to $S$.  Vectors
$\bfv_{h+2},\bfv_{h+3},...,\bfv_{m/2}$ are all orthogonal to $\bfv_h$ so do
not have to change.  We update ${\bfv^+}'$ by
$${\bfv^+}':={\bfv^+}'-\bfv_{h+1}\left (\frac {\bfv_{h+1}\cdot
{\bfv^+}'}{\bfv_{h+1}\cdot \bfv_{h+1}}\right
).$$
Since $\bfv_{h+1}\cdot {\bfv^+}'=2\delta$ and $\bfv_{h+1}\cdot \bfv_{h+1}=2$, this formula is
$${\bfv^+}':={\bfv^+}'-\delta \bfv_{h+1}.$$
Since before modification
${\bfv^+}'=\bfgamma+\delta\cdot \bfp_{m-2h}$, the new
${\bfv^+}'=\bfgamma+\delta\cdot \bfp_{m-2h-2}$.
A similar computation shows that the new ${\bfv^-}'=
-\bfgamma+\delta\cdot \bfp_{m-2h-2}$.
Yet another similar computation shows that $\bfb'$ is now updated to
$\bfp_{m-2h-2}$.
This completes the inductive step.

Since until the end of the last iteration, $||\bfb'||\ge \sqrt
2>\eps$, the algorithm continues running for $m/2$ iterations,
when $||\bfb'||$ drops to 0 and the algorithm terminates. \qed

The authors have a proof, to be included in a future version of this paper, 
that the same example thwarts the LASSO algorithm.
\end{document}